\documentclass[11pt]{article}
\usepackage[backend=biber, maxbibnames=99]{biblatex}
\usepackage{amsmath,amsfonts,amssymb,amsthm,  mathrsfs,mathtools}
\usepackage[margin=1in]{geometry}
\usepackage{hyperref}
\hypersetup{colorlinks}
\usepackage[ruled]{algorithm2e}
\usepackage{subcaption}
\usepackage{xfrac}
\usepackage{enumitem}
\usepackage[T1]{fontenc}

\usepackage{mleftright}

\usepackage{booktabs}
\usepackage{multirow}

\theoremstyle{plain}
\newtheorem{theorem}{Theorem}[section]
\newtheorem{lemma}[theorem]{Lemma}
\newtheorem{corollary}[theorem]{Corollary}

\theoremstyle{definition}
\newtheorem{definition}[theorem]{Definition}

\newtheorem{observation}[theorem]{Observation}
\DeclareMathOperator*{\argmax}{arg\,max}
\DeclareMathOperator*{\argmin}{arg\,min}
\DeclareMathOperator*{\poly}{poly}

\DeclarePairedDelimiter\ceil{\lceil}{\rceil}
\DeclarePairedDelimiter\floor{\lfloor}{\rfloor}
\DeclareMathOperator*{\prot}{prot}

\newcommand{\proba}{\textsc{DemFairCut}}
\newcommand{\probb}{\textsc{IndFairCut}}
\newcommand{\probc}{\textsc{AuxCut}}
\newcommand{\probd}{\textsc{SB-MinCC}}

\begin{document}

\title{Fair Disaster Containment via Graph-Cut Problems}

% The \author macro works with any number of authors. There are two commands
% used to separate the names and addresses of multiple authors: \And and \AND.
%
% Using \And between authors leaves it to LaTeX to determine where to break the
% lines. Using \AND forces a line break at that point. So, if LaTeX puts 3 of 4
% authors names on the first line, and the last on the second line, try using
% \AND instead of \And before the third author name.

% \author{Amy Babay\thanks{University of Pittsburgh.
% Email: \href{mailto:babay@pitt.edu}{babay@pitt.edu}} \and Michael Dinitz\thanks{Johns Hopkins University. 
% Email: \href{mailto:mdinitz@cs.jhu.edu}{mdinitz@cs.jhu.edu}} \and Prathyush Sambaturu\thanks{University of Virginia. Email: \href{mailto:pks6mk@virginia.edu}{pks6mk@virginia.edu}} \and Aravind Srinivasan\thanks{University of Maryland, College Park. 
% Email: \href{mailto:srin@cs.umd.edu}{srin@cs.umd.edu}} \and Leonidas Tsepenekas \thanks{University of Maryland, College Park.  
% Email: \href{mailto:ltsepene@cs.umd.edu}{ltsepene@cs.umd.edu}}
% \and Anil Vullikanti \thanks{University of Virginia.  
% Email: \href{mailto:vsakumar@virginia.edu}{vsakumar@virginia.edu}}
% }

\author{ Michael Dinitz\thanks{Johns Hopkins University. 
Email: \href{mailto:mdinitz@cs.jhu.edu}{mdinitz@cs.jhu.edu}} \and Aravind Srinivasan\thanks{University of Maryland, College Park. 
Email: \href{mailto:srin@cs.umd.edu}{srin@cs.umd.edu}} \and Leonidas Tsepenekas \thanks{University of Maryland, College Park.  
Email: \href{mailto:ltsepene@umd.edu}{ltsepene@umd.edu}}
\and Anil Vullikanti \thanks{University of Virginia.  
Email: \href{mailto:vsakumar@virginia.edu}{vsakumar@virginia.edu}}
}

\date{}

\maketitle

\begin{abstract}
Graph cut problems are fundamental in Combinatorial Optimization, and are a central object of study in both theory and practice. Furthermore, the study of \emph{fairness} in Algorithmic Design and Machine Learning has recently received significant attention, with many different notions proposed and analyzed for a variety of contexts. In this paper we initiate the study of fairness for graph cut problems by giving the first fair definitions for them, and subsequently we demonstrate appropriate algorithmic techniques that yield a rigorous theoretical analysis. Specifically, we incorporate two different notions of fairness, namely \emph{demographic} and \emph{probabilistic individual} fairness, in a particular cut problem that models disaster containment scenarios. Our results include a variety of approximation algorithms with provable theoretical guarantees.
\end{abstract}

\section{Introduction}\label{sec:intro}
Let $G = (V,E)$ be an undirected graph with vertex set $V$ and edge set $E$, where $n = |V|$ and every $e \in E$ has a cost $w_e \in \mathbb{R}_{\geq 0}$. In addition, we are given a designated ``source'' vertex $s \in V$. We are concerned with attempting to mitigate some sort of ``disaster'' that begins at $s$ and infectiously spreads through the network via the edges.  This means that vertices $v \in V$ that are connected to $s$ (i.e., there exists an undirected $s-v$ path in $G$) are at some sort of risk or disadvantage. 

A natural approach to mitigate the aforementioned spread is to remove edges from $G$, in an attempt to disconnect as many vertices of the graph from $s$ as possible. Specifically, if we remove a cut-set or simply cut $F \subseteq E$ from the graph, we denote by $\prot(V,E \setminus F, s)$ the set of vertices in $V$ that are no longer connected to $s$ in $G_F = (V,E\setminus F)$, and hence are \emph{protected} from the infectious process. At a high-level, the edge removal strategy contains the disastrous event within the set $V \setminus \prot(V,E \setminus F, s)$. Observe now that there is a clear trade-off between the cost $w(F)$\footnote{For a vector $\alpha = (\alpha_1, \alpha_2, \hdots, \alpha_k)$ and a subset $X \subseteq \{1,2,\hdots,k\}$, we use $\alpha(X)$ to denote $\sum_{i \in X}\alpha_i$} of the cut $F$ and $|\prot(V,E \setminus F, s)|$, i.e., the more edges we remove the more vertices we may be able to save. 

The aforementioned trade-off naturally leads to the following optimization problem, which we call \emph{Size Bounded Minimum Capacity Cut} or \probd{} for short. Given a graph $G$ with source vertex $s$ and a integer target value $T > 0$, we want to compute a cut $F \subseteq E$ with the minimum possible cost $w(F)$, such that at least $T$ vertices of $V$ are saved in $G_F = (V,E\setminus F)$, i.e. $|\prot(V,E \setminus F, s)| \geq T$. This problem is NP-hard as shown in \cite{hayrapetyan2005}. The work of \cite{svitkina2004} gave a $O(\log^2 n)$-approximation algorithm for \probd{}, while \cite{hayrapetyan2005,eubank-nature} gave constant factor bicriteria algorithms for it, i.e., algorithms that provide solutions that come within a constant factor of the optimal cut cost, but at the same time might not save at least $T$ vertices. 

Inspired by the recent interest revolving around algorithmic fairness, our goal in this paper is to incorporate such ideas in \probd{}, and initiate the discussion of fairness requirements for cuts in graphs. To the best of our knowledge, our work here is the first to combine fairness with this family of problems.

The first notion of fairness that we consider is the widely used \emph{Demographic Fairness} one. The high-level idea behind this definition is that the set of elements that require ``service'' consists of various subsets---say demographic groups---and the solution should equally and fairly treat and represent each of these groups. In our case, if the vertices of the graph belong to different groups, we would like our solution to fairly separate vertices of each of them from the designated node $s$. \emph{In this way, we will avoid outcomes that completely ignore certain groups for the sake of minimizing the objective function}. Hence, we define the following problem. \\

\noindent \proba{}: In addition to a graph $G = (V, E)$ with weights $\{w_e\}_{e \in E}$ and the source $s \in V$, for some integer $\gamma \geq 1$ we are given sets $V_1, V_2, \hdots, V_{\gamma}$ and values $f_1, f_2, \hdots, f_{\gamma}$, such that $\forall h \in [\gamma]$\footnote{We use $[k]$ to denote $\{1,\hdots,k\}$ for some integer $k \geq 1$} we have $V_h \subseteq V$ and $f_h \in (0,1]$. Note that each $v \in V$ may actually belong to multiple sets $V_h$. Letting $n_h = |V_h|$, the goal is to find a cut $F \subseteq E$ with the minimum possible $w(F)$, subject to the constraint that $|V_h \cap \prot(V,E \setminus F,s)| \geq f_h \cdot n_h$ for all $h \in [\gamma]$. \emph{In words, if each $V_h$ is interpreted as a demographic, we want the minimum cost cut under the condition that at least an $f_h$ fraction of the points in $V_h$ are disconnected from $s$ (for all $h$)}.\\

Instantiating this definition with different values of $f_h$ allows us to model a variety of fairness scenarios. For example, setting $f_h = 1/2$ would let us guarantee a solution that protects at least half the vertices of each $V_h$. Alternatively, we can set $f_h$ to be a decreasing function of $n_h/n$, and thus yield a solution that focuses more on protecting smaller demographics. Moreover, notice that \probd{} is a special case of \proba{}, where $\gamma = 1$ (we only have one demographic group) and $f_1 = \frac{T}{n}$. Hence, \proba{} is NP-hard, since \probd{} is already known to be NP-hard.

The second notion of fairness we consider is called \emph{Probabilistic Individual Fairness}, and was first introduced in the context of robust clustering \cite{Harris2019,anegg2020}. According to it, the final solution should not simply be just one solution, but rather a distribution $\mathcal{D}$ over solutions. Then, considering each input element individually, the probability that it will get ``good service'' in a randomly drawn solution from this distribution, should be at most some given (fairness related) parameter. Obviously, sampling from this constructed distribution $\mathcal{D}$ must be achievable in polynomial time, and we call such distributions \emph{efficiently-sampleable}. Under this notion of fairness, \emph{we avoid outcomes that deterministically prevent satisfactory outcomes for certain individuals}. 

Incorporating the above concept of fairness in \probd{}, implies that besides the global guarantee of saving at least $T$ vertices, we also need to provide a stochastic guarantee for each individual vertex, ensuring it that in the final solution it will be disconnected from $s$ with a certain probability. For instance, ensure that each vertex gets disconnected with probability at least $1/2$, and hence no specific vertex enjoys preferential treatment. The formal definition follows.\\

\noindent \probb{}: In addition to a graph $G = (V, E)$ with weights $\{w_e\}_{e \in E}$, a target $T \in \mathbb{N}_{\geq 0}$ and source $s \in V$, for each $v \in V \setminus \{s\}$ we are also given a value $p_v \in [0,1]$. The goal is to find an efficiently-sampleable distribution $\mathcal{D}$ over the cuts $\mathcal{F}(B) = \{F \subseteq E : w(F) \leq B \land |\prot(V,E \setminus F,s)| \geq T\}$, such that $\Pr_{F \sim \mathcal{D}}[v \in \prot(V, E \setminus F, s)] \geq p_v$ for each $v \in V \setminus \{s\}$, and $B$ is the minimum possible.\\

Further, \probd{} is also a special case of \probb{}, since we can always set $p_v = 0$ for all $v \in V \setminus \{s\}$ and make the stochastic constraints void. Hence, \probb{} is also NP-hard.

\begin{observation}
In both problems, we can assume that the disastrous event simultaneously starts from a set of vertices $S$, instead of just a single designated vertex. This assumption is without loss of generality, since $S$ can be merged into a single vertex $s$ (by retaining all edges between $S$ and $V \setminus S$), thus giving an equivalent formulation that matches ours.
\end{observation}

\subsection{Contribution and Outline}\label{sec:contr}
Our main contribution lies in introducing the first fair variants of graph-cut problems, together with approximation algorithms with provable guarantees for them.

In Section \ref{sec:red-trees} we present a technique that is required in our approach for solving \proba{} and \probb{}. The key insight is that we can reduce these problems on general graphs to the same problems on trees, by using a tree embedding result of \cite{Racke08}.

In Section \ref{sec:dem-fair} we address demographic fairness. At first, we provide an $O(\log n)$-approximation algorithm for \proba{} based on dynamic programming. The latter algorithm runs in polynomial time only when the number of groups $\gamma$ is a constant. When $\gamma$ is not a constant and can be any arbitrary value, we develop a different algorithm based on a linear programming relaxation together with a dependent randomized rounding technique. This result yields an $O\mleft(\frac{\log n \log \gamma}{ \epsilon^2 \cdot \min_h f_h  }\mright)$-approximation for any $\epsilon > 0$. However, we mention that the covering guarantee it provides to each demographic $V_h$ is only that at least $(1-\epsilon)f_h n_h$ vertices of it will be saved. Regarding the dependence on $\min_h f_h$, we believe that in realistic fairness related applications the covering fractions $f_h$ should be relatively big, i.e., some constant $f_h = \Omega(1)$, since we care about protecting the vertices in the best way possible. Hence, the approximation ratio of our algorithm can be thought of as $O\mleft(\frac{\log n \log \gamma}{ \epsilon^2 }\mright)$. Finally, we show that even on tree instances, \proba{} with arbitrary $\gamma$ is actually quite hard: it cannot be approximated better than $\Omega(\log \gamma)$. We do this by demonstrating an approximation factor preserving reduction from \textsc{Set Cover}.

In Section \ref{sec:ind-fairness} we provide an  $O(\log n)$-approximation algorithm for \probb{}. The high-level approach of this result relies on the round-or-cut technique developed by \cite{anegg2020}, which we tailor in a way that suits the specific needs of our problem.

Finally, notice that since \probd{} is a special case of \proba{} (with $\gamma = 1$), and also a special case of \probb{} (when $p_v = 0$ for all $v$), our dynamic programming algorithm from Section \ref{sec:dem-fair} and the algorithm of Section \ref{sec:ind-fairness}, both provide a $O(\log n)$-approximation for \probd{}. This constitutes an improvement over the best previously known $O(\log^2 n)$-approximation of \cite{svitkina2004}.

\subsection{Motivating Examples}\label{sec:mot-exmp}

Regarding demographic fairness, consider the following potential application. The vertices of the graph $V$ would correspond to geographic areas across the globe, and an edge $(u,v) \in E$ would denote whether or not there is underlying infrastructure, e.g., highways or airplane routes, that can transport people between areas $u$ and $v$. The disastrous event in this scenario is the spread of a disease in a global health crisis. If an area $u \in V$ is ``infected'', then it is natural to assume that neighboring areas (i.e., areas $v \in V$ with $(u,v) \in E$) can also get infected if we allow people to travel between $u$ and $v$. A central planner will now naturally try to break a set of connections $F \subseteq E$ from the infrastructure graph, such that the total cost $w(F)$ of these actions will be as small as possible, while some guarantee on the number of protected areas $|\prot(V,E\setminus F,s)|$ is also satisfied. The value $w(F)$ can be interpreted as the economic cost of the proposed strategy $F$, e.g., the lost revenue of airline companies resulting from cancelling flights. 

In terms of fairness, we can think of the areas $V$ as coming from $\gamma$ different countries, with $V_h$ being the areas associated with country $h \in [\gamma]$. Then, a fair solution would not tolerate a discrepancy in how many areas are protected across different countries. For example, a fair approach would be to ensure that each country has at least half of its areas protected, since the less ``infected'' areas each country has, the more easily it can keep its local crisis under control.

As far as individual fairness is concerned, consider a computer network facing the spread of a computer virus. In this scenario, we want to minimize the cost of the connections removed, such that the infectious process is kept under control and thus a certain number of users $T$ does not get infected. However, each individual user of the network would arguably prefer to be in the set of protected vertices. Our notion of  individual fairness as studied in \probb{}, will ensure exactly that in a stochastic sense, by using appropriate values $p_v$.

\subsection{Related Work}\label{sec:rel-wrk}

The unfair variant of our problems, i.e., \probd{}, was studied in \cite{hayrapetyan2005, svitkina2004, eubank-nature}. These papers also considered additional versions of \probd{}, where the goal was to maximize $|\prot(V,E \setminus F, s)|$ (equivalently minimize $|V \setminus \prot(V,E \setminus F, s)|$) subject to an upper bound constraint on $w(F)$.

The study of fairness in algorithmic design and machine learning has recently received significant attention. This is mainly due to the realization that the output of standard optimization algorithms can very well lead to solutions that are highly unfair and hurtful for the individuals or the groups involved. Examples of this include racial bias in Airbnb rentals \cite{fairness-1}, gender bias in Google's Ad Settings \cite{fairness-2} and discrimination in housing ads in Facebook \cite{fairness-3}.  There are two reasons why such unfortunate events occur. First, the training datasets may include implicit biases, and hence when algorithms are trained on them, they learn to perpetuate the underlying biases. Second, in many situations, even if the data is completely unbiased, merely optimizing an objective function does not suffice if fairness considerations are at play. In such cases, we must explicitly incorporate fairness constraints in our algorithm design process. Our work here tries to accomplish the latter.

Although algorithmic fairness has not yet been addressed in cut problems, there are other areas such as classification and clustering were examples of fair algorithms are abundant. For example, \cite{Chierichetti2017,bercea2019,bera2020,huang2019,backurs2019,ahmadian2019} consider notions of demographic fairness in clustering, while \cite{brubach2020, anegg2020, Harris2019} focus on notions of individually-fair clustering. In the context of fair classification, one of the most seminal works with significant implications in other fields as well, is the paper of \cite{Dwork2012}. This work studies individual fairness and its interplay with a notion of demographic fairness, namely statistical parity. Some excellent surveys on the topic of algorithmic fairness are \cite{barocas-hardt-narayanan, mehrabi2021}.

\section{Reduction to Tree Instances}\label{sec:red-trees}

In this section we show how both \proba{} and \probb{} can be effectively reduced to solving an appropriate problem on a tree instance. To do this, we use the following lemma.

%where for graph $G=(V,E)$ and any $S \subseteq V$, $\delta(S)$ denotes the set of edges in $E$ with exactly one endpoint in $S$.

\begin{lemma}[\cite{Racke08}] \label{racke-lem}
For any undirected $G=(V,E)$ with edge costs $w_e \geq 0$, we can efficiently construct a collection of trees $T_1 = (V,E_1), ~T_2=(V,E_2), \hdots, ~T_k=(V,E_k)$ with tree $T_i$ having an edge-cost function $w^i : E_i \mapsto \mathbb{R}_{\geq 0}$, and find non-negative multipliers $(\lambda_1, \dots, \lambda_k)$, such that $\sum_{i=1}^k \lambda_i = 1$ and $k = \poly(|V|)$\footnote{Throughout, ``poly" will denote an arbitrary univariate polynomial: its usage in different places could connote different polynomials.}. Further, for any $S \subseteq V$ let $\delta(S)$ be the set of edges in $E$ with exactly one endpoint in $S$, and $\delta_i(S)$ denote the set of edges in $E_i$ with exactly one endpoint in $S$. Then, for any $S \subseteq V$:
\begin{enumerate}
    \item $w(\delta(S)) \leq w^i(\delta_i(S))$ for every $i \in [k]$
    \item $\sum_{i=1}^k \lambda_i w^i(\delta_i(S)) \leq O(\log n) w(\delta(S))$
\end{enumerate}
\end{lemma}

%\footnote{Throughout, ``poly" will denote an arbitrary univariate polynomial: its usage in different places could connote different polynomials.}

\begin{definition}\label{dem-alg}
We call an algorithm for \proba{} \emph{$(\rho,\alpha)$-bicriteria}, if for any given problem instance $\mathcal{I}=\{V,E,s,w,V_1,\hdots, V_{\gamma}, \vec f\}$ with optimal value $OPT_{\mathcal{I}}$, it returns a solution $F$ such that \textbf{1)} $w(F) \leq \rho OPT_{\mathcal{I}}$, and \textbf{2)} $|\prot(V,E \setminus F,s) \cap V_h| \geq \alpha f_h n_h, ~ \forall h \in [\gamma]$. 
\end{definition}

\begin{lemma}\label{dem-tree-red}
If we have a $(\rho,\alpha)$-bicriteria algorithm for \proba{} in trees, we can get a $(\rho \cdot O(\log n),\alpha)$-bicriteria algorithm for \proba{} in general graphs.  
\end{lemma}

\begin{proof}
If $\mathcal{I}=\{V,E,s,w,V_1,\hdots, V_{\gamma}, \vec f\}$ is the general instance, we first apply the result of Lemma~\ref{racke-lem} in order to get a collection of trees $T_1 = (V,E_1), \hdots, ~T_k=(V,E_k)$, where each tree $T_i$ has an associated edge weight function $w^i$. We then use the given algorithm and solve \proba{} in each tree instance $\mathcal{I}_i = \{V,E_i,s,w^i,V_1,\hdots, V_{\gamma}, \vec f\}$, and get a solution $F_i \subseteq E_i$ in return. For the solution $F_i$ we compute for $\mathcal{I}_i$, let $X_i = \prot(V, E_i \setminus F_i, s)$, and note that the properties of the algorithm ensure $|X_i \cap V_h| \geq \alpha f_h n_h, ~\forall h \in [\gamma]$. 

After running the algorithm in each tree instance, we find the tree $T_m$ with $m = \argmin_{i} w(\delta(X_i))$, and we set our solution for the general graph to be $\delta(X_m)$. This means that in our general solution $X_m \subseteq \prot(V, E \setminus \delta(X_m),s)$. Combining this observation with the fact that $|X_m \cap V_h| \geq \alpha f_h n_h$ for all $h \in [\gamma]$, implies that in the solution for the general graph we again satisfy all demographic constraints up to an $\alpha$ violation. We now only have to reason about the cost of $\delta(X_m)$.

Let $X^*$ be the set of vertices not connected to $s$ in the optimal solution of $\mathcal{I}$. If $OPT$ is the value of the latter, then $w(\delta(X^*)) \leq OPT$. Also, since $X^*$ satisfies all $\gamma$ demographic constraints exactly, the set $\delta_i(X^*)$ is a feasible solution for $\mathcal{I}_i$, and $OPT_{\mathcal{I}_i}\leq w^i(\delta_i(X^*))$. Hence, because $\delta_i(X_i) \subseteq F_i$: 
\begin{align}
    w^i(\delta_i(X_i)) \leq \rho \cdot OPT_{\mathcal{I}_i }\leq \rho \cdot w^i(\delta_i(X^*)) \label{dem-tree-aux-1}
\end{align}
Using the definition of $m$ and the first property of the trees from Lemma \ref{racke-lem} gives
\begin{align}
    w(\delta(X_m)) &\leq \sum_{i=1}^k \lambda_i w(\delta(X_i)) \leq \sum_{i=1}^k \lambda_i w^i(\delta_i(X_i)) \label{dem-tree-aux-2}
\end{align}
Combining (\ref{dem-tree-aux-1}), (\ref{dem-tree-aux-2}) and the second property of Lemma \ref{racke-lem} yields
\begin{align}
    w(\delta(X_m)) &\leq \rho \sum_{i=1}^k \lambda_i w^i(\delta_i(X^*)) \leq \rho \cdot O(\log n) \cdot  w(\delta(X^*)) \leq \rho \cdot O(\log n) \cdot  OPT&\qedhere
\end{align}
\end{proof}

Our approach for tackling \probb{} uses as a black-box an algorithm for a new problem, which we call \probc{} and we formally define below. In order to get an algorithm for general instances of \probc{}, we again use a reduction to trees.\\

\noindent \probc{}: We are given an undirected graph $G=(V,E)$, a designated vertex $s \in V$, a budget $B > 0$, and a target value $T \in \mathbb{N}_{\geq 0}$. In addition, each $e \in E$ has a weight $w_e \geq 0$, and each vertex $v \in V \setminus \{s\}$ has a value $a_v \geq 0$. The goal is to find a cut $F$ with $w(F) \leq B$ and $|\prot(V,E \setminus F, s)| \geq T$, that maximizes $a(\prot(V,E \setminus F,s))$.

\begin{definition}\label{aux-alg}
We say that an algorithm is $(1,1,\rho)$-bicriteria for \probc{}, if for any given instance $\mathcal{I} = (V,E,B,T,w,s,a)$ of the problem with optimal value $OPT_{\mathcal{I}}$, it returns a set of edges $F$, such that \textbf{1)} $w(F) \leq \rho B$, \textbf{2)} $|\prot(V,E\setminus F,s)| \geq T$ and \textbf{3)} $a(\prot(V,E\setminus F,s)) \geq OPT_{\mathcal{I}}$.
\end{definition}

\begin{lemma}\label{aux-red}
If we have a $(1,1,\rho)$-bicriteria algorithm for \probc{} in tree instances, we can devise a $(1,1,\rho \cdot O(\log n))$-bicriteria algorithm for \probc{} in general graphs.
\end{lemma}

\begin{proof}
Let $\mathcal{I} = (V,E,B,T,w,s,a)$ be an instance of \probc{} for a general graph. We first apply the result of Lemma \ref{racke-lem} in order to get a collection of trees $T_i=(V,E_i)$ with edge-weight functions $w^i$. Then, for each such tree we create an instance $\mathcal{I}_i = (V,E_i,B\cdot O(\log n),T,w^i,s,a)$, and we use the given bicriteria algorithm to solve \probc{} on it. Let $F_i \subseteq E_i$ the solution we get for $\mathcal{I}_i$, and for notational convenience let again $X_i = \prot(V, E_i \setminus F_i, s)$. After that, we find the tree $T_m$ with $m = \argmax_{i} a(X_i)$, and we set our solution for the general graph to be $\delta(X_m)$. This means that in our general solution we again get $X_m \subseteq \prot(V, E \setminus \delta(X_m), s)$.

At first, because of the properties of the algorithm used on $\mathcal{I}_i$, we have $|X_m| \geq T$, and therefore even in our solution for the general graph we end up saving at least $T$ vertices.

Furthermore, because $\delta_i(X_i) \subseteq F_i$, the properties of the bicriteria algorithm give $w^i(\delta_i(X_i)) \leq \rho \cdot O(\log n) \cdot B$ for every $i$. From the first property in Lemma \ref{racke-lem} we thus get 
\begin{align}
    w(\delta(X_m)) \leq w^m(\delta_m(X_m)) \leq \rho \cdot O(\log n) \cdot B \notag
\end{align}

To conclude we need to show that $a(X_m) \geq OPT_{\mathcal{I}}$, where $OPT_{\mathcal{I}}$ the value of the optimal solution of $\mathcal{I}$. Let also $X^*$ be the set of vertices not connected to $s$ in the optimal solution of $\mathcal{I}$. Since $X^*$ is the optimal such set of vertices, we have $|X^*| \geq T$ and $w(\delta(X^*))\leq B$. Moreover, let $m^* = \argmin_{i}w^i(\delta_i(X^*))$. The definition of $m^*$ and the second property from Lemma \ref{racke-lem} give 
\begin{align}
    w^{m^*}(\delta_{m^*}(X^*)) &\leq \sum_{i=1}^k \lambda_i w^i(\delta_i(X^*)) \leq O(\log n)w(\delta(X^*)) \leq B \cdot O(\log n) \notag
\end{align}
Hence $\delta_{m^*}(X^*)$ is feasible for $\mathcal{I}_{m^*}$ (recall that $|X^*| \geq T$), and since the given algorithm is a $(1,1,\rho)$-bicriteria we get
$
    a(X_m) \geq a(X_{m^*}) \geq OPT_{\mathcal{I}_{m^*}} \geq a(X^*) = OPT_{\mathcal{I}}
$.
\end{proof}

\section{Addressing Demographic Fairness}\label{sec:dem-fair}

In this section we tackle \proba{} and present two algorithms for it. The first works when is $\gamma$ a constant, and is an $O(\log n)$-approximation. The second addresses the case of an arbitrary $\gamma$, and for any $\epsilon > 0$ it is an $\big{(}O\mleft(\frac{\log n \log \gamma}{ \epsilon^2 \cdot \min_h f_h}\mright), 1-\epsilon)$-bicriteria one. 

\subsection{Solving \normalfont{\proba{}} \textbf{for $\gamma = O(1)$}}\label{sec:dem-fair-const}

Given Lemma \ref{dem-tree-red}, we can focus on only solving the problem in tree instances. Specifically, we show that when $\gamma = O(1)$ the problem in trees can be solved optimally via dynamic programming. Without loss of generality, we can also assume that the given tree is rooted at $s$ and it is binary. For details on why this assumption is safe to use, we refer the reader to Lemma $15.18$ from \cite{approx}. Before we describe our approach we need some additional notation. For a vertex $v$, let $\phi_h(v) = 1$ if $v \in V_h$ and $0$ otherwise.

Our dynamic programming algorithm is based on a table $M$, where $M[v,k_1, k_2, \hdots, k_\gamma]$ represents the minimum cost of a cut in the subtree rooted at $v$, so that there are exactly $k_h$ nodes from $V_h$ that are connected to $v$. Let $v_r$ be the right child of $v$, and let $v_{\ell}$ be the left child of $v$. Observe that the optimal solution either cuts neither of the edges from $v$ to its children, just the left edge, just the right edge, or both of the edges. So, we set $M[v,k_1, k_2, \hdots, k_\gamma]$ to the minimum of the following:
\begin{enumerate}
    \item $\min \Big{\{}M[v_{\ell}, k^{\ell}_1, k^{\ell}_2, \hdots, k^{\ell}_\gamma] + M[u_{r}, k^{r}_1, k^{r}_2, \hdots, k^{r}_\gamma]:  k^{\ell}_h + k^r_h + \phi_h(v)= k_h ~\forall h \in [\gamma]\Big{\}}$
    \item $\min \Big{\{} w_{(v, v_{\ell})} + M[v_{r}, k'_1, k'_2, \hdots, k'_\gamma]: k'_h + \phi_h(v) = k_h ~\forall h \in [\gamma] \Big{\}}$ 
    \item $\min \Big{\{}w_{(v, v_r)} + M[v_{\ell}, k'_1, k'_2, \hdots, k'_\gamma]: k'_h + \phi_h(v) = k_h ~\forall h \in [\gamma] \Big{\}}$
    \item $w_{(v, v_r)} + w_{(v, v_{\ell})}$ if $k_h = \phi_h(v)$ for all $h \in [\gamma]$, $+\infty$ otherwise.
\end{enumerate}
The first case above corresponds to cutting neither of the edges $(v,v_r)$, $(v,v_{\ell})$, the second to cutting only $(v,v_{\ell})$, the third to cutting only $(v,v_r)$, and the fourth to cutting both. 

To fill in $M$, we begin by initializing $M[v,\phi_1(v), \phi_2(v), \hdots, \phi_\gamma(v)] = 0$ for all leaves $v$ of the tree, and set all other table entries to $+\infty$. Then we proceed by filling the table bottom-up. There are at most $O(n^{\gamma+1})$ table entries, and to compute each one we need to access at most $2n^{\gamma}$ other ones. Thus, the total runtime is $O(n^{2\gamma+1})$. Finally, in order to find the optimal cut, we look for the minimum entry $M[s, k_1,\hdots,k_\gamma]$, such that $k_h \leq (1-f_h) n_h$ for all $h \in [\gamma]$.

\begin{theorem}\label{const-dem-tree}
When $\gamma$ is a constant, we have an optimal dynamic programming algorithm for \proba{} in trees, running in time $O(n^{2\gamma+1})$.
\end{theorem}

Combining Theorem \ref{const-dem-tree} with Lemma \ref{dem-tree-red}, we see that our approach achieves the following.

\begin{theorem}\label{const-dem-tree-gen}
When $\gamma = O(1)$, we give a $O(\log n)$-approximation algorithm for \proba{}.
\end{theorem}

\subsection{Solving \normalfont{\proba{}} \textbf{for an Arbitrary $\gamma$}}\label{sec:dem-fair-arb}

Given Lemma \ref{dem-tree-red}, we again focus on instances $\mathcal{I}=\{V,E,s,w,V_1,\hdots, V_{\gamma}, \vec f\}$, where the underlying graph $T=(V,E)$ is a tree. Moreover, we can assume without loss of generality that the tree is rooted at $s$. Before we proceed with the description of our algorithm, we need some more notation. For every $v \in V$ let $P(s,v) \subseteq E$ be the unique path from $s$ to $v$ in the tree, and $\ell(v) = |P(s,v)|$. In addition, for every $e = (u,v) \in E$ let $P_e = P(s,r(e))$, with $r(e) = \argmin_{z \in \{u,v\}} \ell(z)$. In words, $P_e$ contains the edges of the path that starts from $s$ and finishes just before reaching $e$. The following linear program (LP) is then a valid relaxation of our problem.
\begin{align}
\min &\sum_{e \in E}w_e \cdot x_e & \label{LP-1} \\
&y_v = \sum_{e \in P(s,v)}x_e ~&\forall v \in V \label{LP-2}\\
&\sum_{v \in V_h}y_v \geq f_h \cdot n_h ~&\forall h \in [\gamma] \label{LP-3}\\
&0\leq y_v, x_e \leq 1 ~&\forall v \in V, e\in E \label{LP-4}
\end{align}

In the integral version of LP (\ref{LP-1})-(\ref{LP-4}), $x_e = 1$ iff edge $e$ is included in the cut. Now notice that because the underlying graph is a tree and the edge weights are non-negative, for any $v \in V$ the optimal solution would not choose more than one edge from $P(s,v)$. Therefore, by constraints (\ref{LP-2}) and (\ref{LP-4}) we see that $y_v = 1$ iff $v$ is separated from $s$ in the optimal outcome. Consequently, constraint (\ref{LP-3}) naturally captures the demographic covering requirements. 

Our approach begins by solving LP (\ref{LP-1})-(\ref{LP-4}) in order to get a fractional solution $x,y$. We then apply the following dependent randomized rounding scheme. We consider the edges of the tree in non-decreasing order of $|P_e|$, and for an edge $e$ for which no other edge in $P_e$ is already chosen for the cut, we remove it with probability $x_e / (1- x(P_e))$ if $x(P_e) < 1$. The latter action is well-defined because for every $e' \in P_e$ we have $|P_{e'}| < |P_e|$, and hence $e'$ is considered before $e$ in the given ordering. Further, if an edge $e$ is chosen to be placed in the cut, then all $v \in V$ with $e \in P(s,v)$ are now disconnected from $s$. In addition, observe that due to the dependent nature of this process, no path $P(s,v)$ will have more than one edge of it in the solution.

Algorithm \ref{alg-1} demonstrates all necessary details of the rounding, with $X_e$ being an indicator random variable denoting whether or not $e$ is included in the solution, and $Y_v$ an indicator random variable that is $1$ iff $v$ is disconnected from $s$ in the final outcome.

\begin{algorithm}[t]
For every $e \in E$ set $X_e \gets 0$, and for all $v \in V$ set $Y_v \gets 0$\;
\For {all $e \in E$ in non-decreasing order of $|P_e|$} {
\If {$x(P_e) < 1$ and $X_{e'} = 0$ for all $e' \in P_e$} {
Set $X_e \gets 1$ with probability $x_e / (1-x(P_e))$\;
\If {$X_e = 1$} {
Set $Y_v \gets 1$ for all $\{v \in V: ~ e \in P( s,v)\}$\;
}
}
}
\caption{Randomized Rounding for LP (\ref{LP-1})-(\ref{LP-4})}\label{alg-1}
\end{algorithm}

\begin{lemma}\label{valid-dist}
When we randomly decide to include $e \in E$ in the cut, we do so with a valid probability.
\end{lemma}

\begin{proof}
Let $e = (u,v)$, and without loss of generality assume $l(u) < l(v)$. This means that $P_e = P(s,u)$ and $P(s,v) = P(s,u) \cup \{e\}$. In addition, to consider a randomized decision for $e$ we should also have $x(P_e) < 1$. Using constraints (\ref{LP-2}) and (\ref{LP-4}) for $v$ we therefore get:
\begin{align}
    x_e + \sum_{e' \in P_e}x_{e'} \leq 1  \implies  \frac{x_e}{1-x(P_e)} \leq 1 \notag &\qedhere
\end{align}
\end{proof}

\begin{lemma}\label{removal-prob}
For every $e \in E$ and $v \in V $, we have $\Pr[X_e = 1] = x_e$ and $\Pr[Y_v = 1] = y_v$.
\end{lemma}

\begin{proof}
Let us begin with an $e \in E$ for which we never made a random decision because $x(P_e) \geq 1$, and hence $X_e = 0$. If $e = (u,v)$ with $l(u) < l(v)$, then $P_e = P(s,u)$ and $P(s,v) = P(s,u) \cup \{e\}$. Because of constraints (\ref{LP-2}) and (\ref{LP-4}) for $u$ we first get $x(P_e) = 1$. Therefore, constraints (\ref{LP-2}) and (\ref{LP-4}) applied this time for $v$ yield $x_e = 0$, which indeed gives $\Pr[X_e = 1] = x_e$. 

Now let us consider an edge $e$ with $x(P_e) < 1$. Because for each $e' \in P_e$ we have $P_{e'} \subset P_e$, we also get $x(P_{e'}) < 1$. The latter means that for all other edges in $P_e$ a random decision potentially takes place. Furthermore, analysis of the algorithm's actions shows that $\Pr[X_e = 1]$ is equal to
\begin{align}
    &\Pr[X_e = 1 ~|~ X_{e'} = 0 ~\forall e' \in P_e] \cdot \Pr[X_{e'} = 0 ~\forall e' \in P_e] \notag \\
    &= \frac{x_e}{1-\sum_{e' \in P_e}x_{e'}} \prod_{e' \in P_e}\Big{(}1 - \frac{x_{e'}}{1-\sum_{e'' \in P_{e'}}x_{e''}}\Big{)} \label{removal-prob-aux} 
\end{align}
Let $e_1, \hdots, e_m$ the edges of $P_e$ in increasing order of $|P_{e_j}|$. Then because $P_{e_j} = \{e_{j'} ~|~ j' < j\}$, expression (\ref{removal-prob-aux}) can be rewritten as a telescopic product of fractions:
\begin{align}
 \frac{x_e}{1-\sum^m_{j=1}x_{e_j}}\prod^m_{j=1}\Big{(}1- \frac{x_{e_j}}{1-\sum^{j-1}_{i=1}x_{e_i}} \Big{)} = x_e \notag
\end{align}

As for a vertex $v \in V$, we have $\Pr[Y_v = 1] = \Pr[\exists e \in P(s,v): X_e = 1]$ because there is a unique path from $s$ to it. Moreover, since our rounding will never put more than one edges of $P(s,v)$ in the cut, for all $S \subseteq P(s,v)$ with $|S| \geq 2$ we get $\Pr[X_e = 1, \forall e \in S] = 0$. Hence, by the inclusion-exclusion principle $\Pr[\exists e \in P(s,v): X_e = 1] = \sum_{e \in P(s,v)}\Pr[X_e = 1] = \sum_{e \in P(s,v)}x_e = y_v$, where the last equality follows from constraint (\ref{LP-2}).
\end{proof}

We will now analyze the satisfaction of the coverage constraints for the different demographics. If $S_h$ is the number of vertices from $V_h$ that are not connected to $s$ in the solution, we see that $S_h = \sum_{v \in V_h}Y_v$. Using Lemma \ref{removal-prob} and constraint (\ref{LP-3}) gives $\mathbb{E}[S_h] \geq f_h n_h$. We thus need to calculate how much can $S_h$ deviate from $\mathbb{E}[S_h]$. For that we will need the following two lemmas.

\begin{lemma}\label{lem-janson}
\cite{Janson97} Let $Z_1,\hdots,Z_m$ be Bernoulli random variables, where $\Pr[Z_i = 1]=z_i$ for all $i \in [m]$. Let $\Gamma$ be the dependency graph on the $Z_i$. For $i \neq j$, $Z_i$ and $Z_j$ are dependent if there exists an edge between them in $\Gamma$, and we denote that as $i \sim j$. Let also $Z = \sum^m_{i=1}Z_i$, $\mu = \mathbb{E}[Z]$, $\Delta = \sum_{\{i,j\}:i \sim j}\Pr[Z_i = Z_j = 1]$, $\delta_i = \sum_{j \sim i}z_j$ and $\delta = \max_i \delta_i$. Then for any $\epsilon \in [0,1]$ 
\begin{align}
\Pr[Z \leq (1-\epsilon)\mu] \leq \exp \Big{(} -\min \Big{(} \frac{\epsilon^2 \cdot \mu^2}{8\Delta + 2\mu}, \frac{\epsilon \cdot \mu}{6\delta}\Big{)}\Big{)}\notag
\end{align}
\end{lemma}

\begin{lemma}\label{aux-lem}
For every $m \in \mathbb{N}_{>0}$ and some sequence of non-negative numbers $a_1, a_2, \hdots$ we have:
$$\sum^{m-1}_{i=1}(m-i)a_i \leq m\sum^{m}_{i=1}a_i$$
\end{lemma}

\begin{proof}
We prove the statement via induction on $m$. For $m=1$ it is trivial. Suppose that the lemma holds up to some $m = k$. We then prove it for $m=k+1$:
\begin{align}
    \sum^{k+1-1}_{i=1}(k+1-i)a_i &= \sum^{k}_{i=1}\Big{(}(k-i)a_i + a_i\Big{)} = \sum^{k}_{i=1}(k-i)a_i + \sum^{k}_{i=1}a_i \notag \\& = \sum^{k-1}_{i=1}(k-i)a_i + \sum^{k}_{i=1}a_i \leq k\sum^{k}_{i=1}a_i + \sum^{k}_{i=1}a_i \notag \\ &\leq (k + 1)\sum^{k}_{i=1}a_i \leq (k + 1)\sum^{k+1}_{i=1}a_i \notag
\end{align}
The first inequality uses the inductive hypothesis, while the last one the fact that $a_{k+1} \geq 0$.
\end{proof}

\begin{lemma}\label{iter-bound}
For all $h \in [\gamma]$ and any $\epsilon \in [0,1]$, we have $\Pr[S_h \leq (1-\epsilon)\mathbb{E}[S_h]] \leq e^\frac{-\epsilon^2 \cdot f_h}{10}$.
\end{lemma}

\begin{proof}
Due to Lemma \ref{removal-prob}, the random variables $Y_v$ for $v \in V_h$ are Bernoulli with $\Pr[Y_v = 1] = y_v$. Because of the tree structure they are also to some extent dependent. Our goal here is to apply Lemma \ref{lem-janson} for $S_h$, and towards that end we need to upper bound the dependency factors $\delta, \Delta$. Since we do not know exactly the underlying dependency graph $\Gamma$, in what follows we assume that all pairs $Y_v, Y_{v'}$ are dependent. We begin by upper-bounding the parameter $\Delta$ of Lemma \ref{lem-janson}.
\begin{align}
    \Delta &\leq \sum_{\{v,v'\} \in V_h}\Pr[Y_v = Y_{v'} = 1] \notag \\ &\leq \sum_{\{v,v'\} \in V_h}\min (\Pr[Y_v = 1], \Pr[Y_{v'} = 1])\notag \\ &= \sum_{\{v,v'\} \in V_h}\min (y_v, y_{v'}) \notag
\end{align}
Now let $a_1, a_2, ..., a_{n_h}$ be the values $y_v$ for all $v \in V_h$ in non-decreasing order. Then we have:
\begin{align}
  \sum_{\{v,v'\} \in V_h}\min (y_v, y_{v'}) &= \sum^{n_h-1}_{i=1}(n_h - i)a_i \leq n_h\sum^{n_h}_{i=1}a_i = n_h \cdot \mathbb{E}[S_h] \notag 
\end{align}
To get the first inequality we used Lemma \ref{aux-lem}. Therefore, we get $\Delta \leq n_h \cdot \mathbb{E}[S_h]$. Moreover, a straightforward upper bound for each $\delta_v$ is $\delta_v \leq \sum_{u \in V_h}y_u = \mathbb{E}[S_h]$. Thus, $\delta \leq \mathbb{E}[S_h]$. Finally, we also need bounds for the following two quantities, where $\mu = \mathbb{E}[S_h]$:
\begin{align}
    \frac{\epsilon^2 \cdot \mu^2}{8\Delta + 2\mu} &\geq \frac{\epsilon^2 \cdot \mu^2}{8\mu \cdot n_h + 2\mu} = \frac{\epsilon^2 \cdot \mu}{8n_h + 2}\geq \frac{\epsilon^2 \cdot n_h \cdot f_h}{8n_h + 2} \geq \frac{\epsilon^2 \cdot f_h}{10} \notag 
\end{align}
\begin{align}
    \frac{\epsilon \cdot \mu}{6\delta} \geq \frac{\epsilon \cdot \mu}{6\mu} = \frac{\epsilon}{6} \notag
\end{align}

Since $\frac{\epsilon}{6} \geq \frac{\epsilon^2 \cdot f_h}{10}$ for any $\epsilon, f_h \in [0,1]$, Lemma~\ref{lem-janson} immediately gives the desired bound.
\end{proof}

To conclude, for some constant $\beta \geq 2$ we repeat Algorithm \ref{alg-1} independently $N= \frac{10\log \gamma^\beta}{\epsilon^2 \cdot \min_{h}f_h}$ times, and in each run $t$ of it (with $t \in [N]$) we compute a set of edges $F_t$ that are chosen to be removed. Our final solution is set to be $F = \bigcup_t F_t$. Then we have the following.

\begin{theorem}\label{dem-overall-bound}
For \proba{} in trees and any $\epsilon \in (0,1)$, we give an $\mleft(O\mleft(\frac{\log \gamma}{ \epsilon^2 \min_h f_h}\mright), 1-\epsilon \mright)$-bicriteria algorithm that runs in expected polynomial time.
\end{theorem}

\begin{proof}
Focus on a specific demographic $h$, and let $S^t_h$ the random variable denoting the number of nodes of $V_h$ separated from $s$ in $(V,E \setminus F_t)$. By Lemma \ref{iter-bound} and the independent nature of the runs:
\begin{align}
\Pr\big{[}S^t_h \leq (1-\epsilon)\mathbb{E}[S^t_h], ~\forall t\big{]} \leq e^\frac{-\epsilon^2 \cdot N \cdot f_h}{10} \leq \frac{1}{\gamma^{\beta}} \notag
\end{align}
Thus, because $\mathbb{E}[S^t_h] \geq f_h  n_h$ for all $t$, we have $$\Pr\Big{[}\big{|}V_h \cap \prot(V,E \setminus F,s)\big{|} \geq (1-\epsilon)f_h  n_h\big{]} \geq \Pr\big{[}\exists t: S^t_h > (1-\epsilon)\mathbb{E}[S^t_h]\Big{]} \geq 1-\frac{1}{\gamma^{\beta}}$$
A union bound over all demographics would finally give $$\Pr\Big{[}\big{|}V_h \cap \prot(V,E \setminus F,s)\big{|} \geq (1-\epsilon)f_h  n_h, ~ \forall h \in [\gamma]\Big{]} \geq 1-\frac{1}{\gamma^{\beta-1}}$$

By Lemma \ref{removal-prob}, in each run an edge $e$ gets removed with probability $x_e$. Hence, with a union bound over all runs, the probability that $e$ gets removed is at most $N x_e$. Therefore, the total expected cost of our algorithm is $N \sum_{e \in E}w_e x_e$, and since LP (\ref{LP-1})-(\ref{LP-4}) is a valid relaxation of the problem, we immediately get the desired approximation ratio on expectation. By Markov's inequality we can further prove that with probability at most $\frac{1}{c}$, we get a final cut of cost more than $cN \sum_{e \in E}w_e x_e$ for some constant $c > 1$.

Thus, with constant probability our algorithm satisfies both the ratio of $O(\frac{\log \gamma}{ \epsilon^2 \min_h f_h})$, and the $1-\epsilon$ approximate satisfaction of the demographic constraints (specifically we fail to satisfy both of the above with probability at most $1/\gamma^{\beta -1} + 1/c$). Hence, repeating the whole process an expected logarithmic number of times, guarantees that we hit both targets deterministically.
\end{proof}

By combining Theorem \ref{dem-overall-bound} and Lemma \ref{dem-tree-red}, we see that our approach achieves the following.

\begin{theorem}\label{dem-overall-bound-gen}
For any given constant $\epsilon \in (0,1)$, we provide an $\mleft(O\mleft(\frac{\log n \log \gamma}{ \epsilon^2 \cdot \min_h f_h}\mright), 1-\epsilon \mright)$-bicriteria algorithm for \proba{}, which also runs in expected polynomial time.
\end{theorem}

\subsubsection{Hardness of \normalfont{\proba{}} \textbf{with Arbitrary $\gamma$}}

Here we show that even in tree instances, \proba{} with arbitrary $\gamma$ is hard. Specifically, we use a reduction from \textsc{Set Cover}.\\

\noindent \textsc{Set Cover}: We are given a universe of elements $U$ and a collection of $m$ sets $\{S_1, S_2, \hdots, S_m\}$, where $S_i \subseteq U$ for every $i \in [m]$. The goal is to find $C \subseteq [m]$, such that $\bigcup_{i \in C}S_i = U$ and $|C|$ is minimized.

\begin{theorem}[\cite{dinur14}]\label{set-cover}
It is NP-hard to approximate Set Cover instances of universe size $n$ and $m \leq \poly(n)$ sets within a factor better than $\ln n$.
\end{theorem}

This allows us to prove the following theorem.

\begin{theorem}
It is NP-hard to approximate \proba{} with arbitrary $\gamma$ on tree instances within a factor better than $\ln \gamma$.
\end{theorem}

\begin{proof}
Suppose that we are given an instance of \textsc{Set Cover}. We create an instance of \proba{} as follows. For every set $S_i$ we create a vertex $v_i$. For every element $e \in U$ we create a demographic group $V_e = \{v_i ~|~ e \in S_i\}$. We set the covering requirement of the group $V_e$ to be $1/|V_e|$, % each demographic to be $1$, 
i.e., we want our solution to protect at least $|V_e| \cdot (1/|V_e|) = 1$ vertex from each $V_e$. Finally, we add the designated vertex $s$ to the graph, and create edges $(s,v_i)$ for every $v_i$. Note that the resulting graph is a tree.

Now consider the optimal \textsc{Set Cover} solution $C^*$. We claim that the set of edges $\{(s,v_i) ~|~ i \in C^*\}$ is a feasible solution for the constructed instance of \proba{}. Take any demographic $V_e$ for $e \in U$. Because $C^*$ is a feasible \textsc{Set Cover} solution, it contains at least one $S_j$ with $e \in S_j$. Therefore, we are going to include the edge $(s,v_j)$ to our graph solution, and the vertex $v_j$ from the group $V_e$ is going to be protected. Finally, see that $|C^*| = |\{(s,v_i) ~|~ i \in C^*\}|$, and hence the cost of the optimal solution for the \proba{} instance, say $F^*$, is at most $|C^*|$.

Now we argue that any solution $F$ to the \proba{} instance yields a feasible solution $C_F$ for the \textsc{Set Cover} instance with $|F| = |C_F|$. Simply take $C_F = \{i \in [m] ~|~ (s,v_i) \in F\}$. It is clear that $|F| = |C_F|$. Now consider each $e \in U$. Since $F$ is feasible for \proba{}, at least one vertex $v_i \in V_e$ will be separated from $s$, and thus $(s,v_i) \in F$. Hence for that vertex $v_i$ we have $e \in S_i$ by construction. Therefore, $e$ is covered by $C_F$.

Suppose now that for some $\epsilon > 0$ we have an $(1-\epsilon) \ln \gamma$-approximation algorithm for \proba{} on trees. Then given an instance of \textsc{Set Cover}, we first construct the instance of \proba{} given by the above reduction and then run the given algorithm on that instance to get a solution $F$. Then, as discussed, we construct the corresponding Set Cover solution $C_F$, with $|F| = |C_F|$. By all the previous arguments we have $|C_F| = |F| \leq ((1-\epsilon) \ln \gamma) |F^*| \leq ((1-\epsilon) \ln |U|)|C^*|$. This contradicts Theorem~\ref{set-cover}.
\end{proof}

At a high-level, the previous theorem says that the best we can achieve for \proba{} in trees is an approximation ratio of $\Omega(\log \gamma)$. Trivially this implies the following corollary.

\begin{corollary}
Unless P$=$NP, the best approximation ratio we can achieve for general instances of \proba{} with arbitrary $\gamma$ is $\Omega(\log \gamma)$.
\end{corollary}

\section{Addressing Individual Fairness}\label{sec:ind-fairness}

The purpose of this section is to provide an algorithm for \probb{}. To do so, we begin by giving a dynamic programming bicriteria algorithm for \probc{} on tree instances, which according to Lemma \ref{aux-red} implies an algorithm for \probc{} in general graphs. Subsequently, we show how the general graph algorithm can be incorporated in the round-or-cut framework of \cite{anegg2020}, and in this way we get as our final result a $O(\log n)$-approximation for \probb{}. 

At this point, we have to mention that the LP-based approach of Section \ref{sec:dem-fair-arb} can also be applied here (by adding the extra constraint $y_v \geq p_v$ in LP (\ref{LP-1})-(\ref{LP-4})), yielding the same approximation ratio of $O(\log n)$. However, such an approach would unavoidably lead to a bicriteria algorithm, since it will produce a solution that saves at least $(1-\epsilon)T$ vertices. On the other hand, the algorithm we present in what follows is a true approximation for \probb{}.

\subsection{A $(1,1,O(\log n))$-Bicriteria Algorithm for \normalfont{\probc{}}}\label{sec:sbcc}

Suppose we have an instance $\mathcal{I}=(V,E,B,T,s,w,a)$ of \probc{}. Given Lemma \ref{aux-red}, we focus on $G=(V,E)$ being a tree and present a dynamic programming algorithm for \probc{} in trees. 

Without loss of generality, we can assume that the tree is rooted at $s$ and is binary (see Lemma $15.18$ from \cite{approx}). Our algorithm tries to find a cut $F \subseteq E$ that minimizes $a(V \setminus \prot(V,E \setminus F,s))$ subject to $w(F) \leq B$ and $|\prot(V,E \setminus F, s)| \geq T$. Note that when we can compute a solution of optimal value to this minimization problem, minimizing $a(V \setminus \prot(V,E \setminus F,s))$ is equivalent to maximizing $a(\prot(V,E \setminus F,s))$. Therefore, the version of the problem we solve here is equivalent to the definition of \probc{} as given in Section \ref{sec:red-trees}.

Our approach relies on a table $A$. For every $v \in V$ let $T_v \subseteq V$ and $E_v \subseteq E$ be the vertices and the edges of the subtree that is rooted at $v$ (with $v$ included in $T_v$). Then, the entry $A[v,W,k]$ would represent the minimum possible $a(T_v \setminus \prot(T_v, E_v \setminus F_v, v))$, for any cut $F_v \subseteq E_v$ with $w(F_v) = W$ and $|T_v \setminus \prot(T_v, E_v \setminus F_v, v)| = k$ (see that the vertices of $T_v$ connected to $v$ in this cut are those in $T_v \setminus \prot(T_v, E_v \setminus F_v, v)$). Let also $v_r$ be the right child of $v$, and let $v_{\ell}$ be the left child of $v$. The optimal solution of $\mathcal{I}$ either cuts none of the edges from $v$ to its children, just the left edge, just the right edge, or both edges. So we just have to set $A[v,W,k]$ to the minimum of the following:
\begin{enumerate}
    \item $\min\Big{\{}A[v_{\ell}, W_{\ell}, k_{\ell}] + A[v_r, W_r, k_r]  + a_v : W_{\ell} + W_r = W \text{ and } k_{\ell} + k_r + 1= k\Big{\}}$
    \item $A[v_r, W-w_{(v,v_{\ell})}, k-1] +  a_v \text{ if } W \geq w_{(v,v_{\ell})} \text{ and } k > 1, +\infty \text{ otherwise}$
    \item $A[v_\ell, W-w_{(v,v_{r})}, k-1] +  a_v \text{ if } W \geq w_{(v,v_{r})} \text{ and } k > 1, +\infty \text{ otherwise}$
    \item $a_v$ if $w_{(v,v_{\ell})} + w_{(v,v_{r})}=W$ and $k=1$, $+\infty$ otherwise
\end{enumerate}

The first case above corresponds to cutting neither of the edges $(v,v_r)$, $(v,v_{\ell})$, the second to cutting only $(v,v_{\ell})$, the third to cutting only $(v,v_r)$, and the fourth to cutting both. 

To fill in $A$, we begin by initializing $A[v,0,1] = a_v$ for all leaves $v$ of the tree, and all other entries to $+\infty$. Then we proceed by filling the table bottom-up. Assuming that the edge weights are integers, we see that $A$ has $n^2 B$ entries, and in order to fill each of them, we need access to at most $2nB$ other entries. Hence, in total our approach requires $O(n^3 B^2)$ time. Finally, in order to find the optimal cut, we look for the minimum entry $A[s, W, k]$, such that $W \leq B$ and $k \leq n - T$.

\begin{corollary}\label{aux-cor}
When the edge weights are integers and $B = \poly(n)$, we can efficiently find an optimal solution of \probc{} in tree instances. 
\end{corollary}

To make sure the edge weights are integers and $B$ is polynomially bounded, we use a standard discretization trick before running the dynamic program \cite{svitkina2004}. Specifically, for any $\epsilon > 0$, let $\lambda = \frac{\ceil*{m/\epsilon}}{B}$, where $m = |E|$. Then for each edge $e \in E$ create a new weight $w'_e = \floor*{\lambda w_e}$. Also, set $B' = \lambda B = \ceil*{m/\epsilon}$. Notice now that all new edge weights are integers and that $B'$ is polynomial in $n$. Further, using these new values we create a new instance $\mathcal{I}'=(V,E,B',T,s,w',a)$ of \probc{}. It is easy to see that if there is a solution of edge-cost $B$ for $\mathcal{I}$, then this solution has edge-cost $B'$ in $\mathcal{I}'$. In addition, for every solution of $\mathcal{I}'$ whose edge-cost is at most $B'$, its edge-cost in $\mathcal{I}$ is at most $(1+\epsilon)B$. Combining this with Corollary \ref{aux-cor} gives the following.

\begin{corollary}\label{aux-cor-2}
Our approach provides a $(1,1,1+\epsilon)$-bicriteria algorithm for \probc{} in trees.
\end{corollary}

Finally, by Corollary \ref{aux-cor-2}, Lemma \ref{aux-red} and the fact that $\epsilon$ is a constant, we get:

\begin{theorem}\label{aux-prob-thm}
Our approach provides a $(1,1,O(\log n))$-bicriteria algorithm for \probc{}.
\end{theorem}

\subsection{A Round-or-Cut Solution for \normalfont{\probb{}}}

Suppose we are given an instance $\mathcal{I} = (V,E,T,s,w,\vec p)$ of \probb{} with optimal value $OPT_{\mathcal{I}}$. For any value $B\geq 0$, let $\mathcal{F}(B) = \{F \subseteq E: w(F) \leq B \text{ and } |\prot(V,E \setminus F,s)| \geq T\}$. In the rest of the section we demonstrate a process, which given $\mathcal{I}$ and a target value $B \geq 0$, operates as follows. It either returns an efficiently-sampleable distribution $\mathcal{D}$ over the cuts in the set $\mathcal{F}(O(\log n)B)$ such that $\Pr_{F \sim \mathcal{D}}[v \in \prot(V,E \setminus F,s)] \geq p_v$ for every $v \in V \setminus \{s\}$, or returns ``INFEASIBLE''. If the latter happens, then it is guaranteed that $B < OPT_{\mathcal{I}}$. 

Using the above process in a bisection search with step $(1+\epsilon)$ over the range $[0,w(E)]$, we can efficiently compute a value $B' \leq (1+\epsilon)OPT_{\mathcal{I}}$, such that the process will not return ``INFEASIBLE'' for $B'$. This will actually yield an efficiently-sampleable distribution over $\mathcal{F}(O(\log n)B')$ that satisfies the stochastic constraints for all vertices. Hence, we get our final result. 

\begin{theorem}\label{ind-final}
For any $\epsilon > 0$ and instance $\mathcal{I}$ with optimal value $OPT_{\mathcal I}$, we construct an efficiently sampleable distribution $\mathcal{D}$ over $\mathcal{F}(O(\log n)(1+\epsilon)OPT_{\mathcal{I}})$, such that $\Pr_{F \sim \mathcal{D}}[v \in \prot(V,E \setminus F,s)] \geq p_v$ for every $v \in V \setminus \{s\}$. Moreover, the runtime of our approach is $\poly(n^{1/\epsilon})$. 
\end{theorem}

Therefore, since for our final result the aforementioned process is all that is required, we start describing its details. Notice now that for a given target value $B$, we are basically interested in verifying whether or not there is a feasible solution to $\mathcal{I}$ with edge-cost at most $B$. Hence, consider the following exponential-sized linear program, which we call PLP$(B)$.

\begin{minipage}{0.5\textwidth}
\begin{align}
&& \textbf{PLP(B)}  \notag \\
\hline
\min &~~0 &\notag \\
&\smashoperator[l]{\sum_{\substack{F \in \mathcal{F}(B): \\ v \in \prot(V,E \setminus F,s)}}} x_F \geq p_v &\forall v \in V \setminus \{s\} \notag\\
&\smashoperator[l]{\sum_{F \in \mathcal{F}(B)}} x_F =1& \notag\\
&0\leq x_F \leq 1 ~&\forall F \in \mathcal{F}(B)\notag
\end{align}
\end{minipage} \vline \hfill
\begin{minipage}{0.5\textwidth}
\begin{align}
&& \textbf{DLP(B)}  \notag \\
\hline
\max &\smashoperator[l]{\sum_{v \in V \setminus \{s\}}}p_v \cdot y_v - \mu\notag &\\
&\smashoperator[l]{\sum_{v \in \prot(V,E\setminus F,s)}} y_v \leq \mu ~&\forall F \in \mathcal{F}(B) \notag\\
&0 \leq y_v ~&\forall v \in V \notag \\
&\mu \in \mathbb{R}&\notag \\ \notag 
\end{align}
\end{minipage}
\\

If we interpret $x_F$ as the probability of choosing the cut $F$ from $\mathcal{F}(B)$, we see that $B$ yields a feasible solution iff PLP$(B)$ is feasible. This is because the first LP constraint captures the fairness requirements, and the second LP constraint the fact that the resulting solution should be a distribution over $\mathcal{F}(B)$. In addition, if PLP$(B)$ is feasible, then there are only $n$ values $x_F$ with $x_F > 0$ (see Lemma 9 in \cite{karloff}), and hence the resulting distribution is efficiently-sampleable. Another important observation is that if PLP$(B)$ is feasible, then clearly its optimal value is $0$.

However, since solving PLP$(B)$ is not doable in polynomial time, we focus on its dual, which we call DLP$(B)$ and we present next to the primal LP.

Here note that DLP$(B)$ is always feasible (e.g., set all variables to $0$), and by LP duality DLP$(B)$ has an optimal value of $0$ iff PLP$(B)$ is feasible. Further, see that DLP$(B)$ is scale-invariant. In other words, if it has a feasible solution $(y',\mu')$ with strictly positive objective value, then DLP$(B)$ is unbounded because $(t y', t \mu')$ will also be feasible for any $t > 0$. Consider now the following polytope that contains all feasible solutions of DLP$(B)$ of objective value at least $1$.
\begin{align}
    Q(B) = \Big{\{}(y,\mu)\in \mathbb{R}^{n-1}_{\geq 0} \times \mathbb{R}: \sum_{v \in V \setminus \{s\}}p_v  y_v \geq \mu + 1 \land y(\prot(V,E\setminus F,s)) \leq \mu, ~\forall F \in \mathcal{F}(B)\Big{\}} \notag
\end{align}
Based on the previous discussion we make the following very crucial observation.
\begin{observation}\label{pol-feas}
PLP$(B)$ is feasible iff $Q(B) = \emptyset$.
\end{observation}

Using the algorithm of Section \ref{sec:sbcc} we prove the following vital theorem.

\begin{theorem}\label{ind-main}
There exists a poly-time algorithm that given a point $(y,\mu) \in \mathbb{R}^{n-1}_{\geq 0} \times \mathbb{R}$ satisfying $\sum_{v \in V \setminus \{s\}}p_v \cdot y_v \geq \mu + 1$, it either verifies that $(y,\mu) \in Q(B)$, or outputs a set $F \in \mathcal{F}(O(\log n)B)$ such that $\sum_{v \in \prot(V,E\setminus F,s)} y_v >\mu$.
\end{theorem}

\begin{proof}
We begin by constructing an instance $\mathcal{I}_{\text{aux}}=(V,E,B,T,s,w,y)$ of \probc{}, where the vertex weights correspond to the $y$ values. Then, we run the algorithm of Section~\ref{sec:sbcc} on $\mathcal{I}_{\text{aux}}$. Suppose now that $F \subseteq E$ is the solution returned by the algorithm, for which by Theorem \ref{aux-prob-thm} we have $w(F) \leq O(\log n)B$ and $|\prot(V,E \setminus F,s)| \geq T$. If  $y(\prot(V,E \setminus F,s)) > \mu$, then we return $F$ as our answer, because we are guaranteed to have $F \in \mathcal{F}(O(\log n) B)$. If on the other hand $y(\prot(V,E \setminus F,s)) \leq \mu$, then all $F' \in \mathcal{F}(B)$ have $y(\prot(V,E \setminus F',s)) \leq \mu$, because the properties of the Section \ref{sec:sbcc} algorithm ensure that $y(\prot(V,E \setminus F,s)) \geq y(\prot(V,E \setminus F',s))$. The latter immediately indicates that $(y,\mu) \in Q(B)$.
\end{proof}

Given the existence of an algorithm like the one described in Theorem \ref{ind-main}, \cite{anegg2020} prove that with a round-or-cut approach we can either show that $Q(B) \neq \emptyset$ or that $Q(O(\log n)B) = \emptyset$.  If $Q(B) \neq \emptyset$, then by Observation \ref{pol-feas} we can infer $B < OPT_{\mathcal{I}}$ and return ``INFEASIBLE''. If on the other hand $Q(O(\log n)B) = \emptyset$, then again by Observation \ref{pol-feas} we know that PLP$(O(\log n)B)$ is feasible. Furthermore, in the latter case the framework of \cite{anegg2020} provides a set $\mathcal{F}' \subseteq \mathcal{F}(O(\log n)B)$ with polynomial size, for which the following (poly-sized) LP is feasible.
\begin{align}
\min &~~0 &\notag \\
&\smashoperator[l]{\sum_{\substack{F \in \mathcal{F}': \\ v \in \prot(V,E \setminus F,s)}}} x_F \geq p_v &\forall v \in V \setminus \{s\} \notag\\
&\smashoperator[l]{\sum_{F \in \mathcal{F}'}} x_F =1& \notag\\
&0\leq x_F \leq 1 ~&\forall F \in \mathcal{F}'\notag
\end{align}
Finally, since the above can be efficiently solved, we obtain an efficiently-sampleable distribution $\mathcal{D}$ over $\mathcal{F}(O(\log n)B)$, such that $\Pr_{F \sim \mathcal{D}}[v \in \prot(V,E \setminus F,s)] \geq p_v$ for all $v \in V \setminus \{s\}$.

\section*{Acknowledgements}

Michael Dinitz was supported by NSF award CCF-1909111. Aravind Srinivasan was supported in part by NSF awards CCF-1422569, CCF-1749864, and CCF-1918749, as well as research awards from Adobe, Amazon, and Google. Leonidas Tsepenekas was  supported in part by NSF awards CCF-1749864 and CCF-1918749, and by research awards from Amazon and Google. Anil Vullikanti's work was partially supported by NSF awards IIS-1931628, CCF-1918656, and IIS-1955797, and NIH award R01GM109718.

\printbibliography

\end{document}